\documentclass[reprint,amsmath,amssymb,aps,pra]{revtex4-2}

\usepackage{dcolumn}
\usepackage{bm}

\usepackage{relsize}
\usepackage[utf8]{inputenc}
\usepackage[english]{babel}
\usepackage[T1]{fontenc}
\usepackage{hyperref}
\usepackage{graphicx}

\usepackage{amsthm}
\usepackage{amsmath} 
\usepackage{amssymb}
\usepackage[colorinlistoftodos]{todonotes}
 
\newtheorem{theorem}{Theorem}
\newtheorem{corollary}[theorem]{Corollary}
\newtheorem{lemma}[theorem]{Lemma}
\newtheorem{definition}[theorem]{Definition}
\newtheorem{remark}[theorem]{Remark}
\newtheorem{example}[theorem]{Example}

\begin{document}

\preprint{APS/123-QED}

\title{Transforming Collections of Pauli Operators into Equivalent Collections of Pauli Operators over Minimal Registers}

\author{Lane G. Gunderman}
 \email{lgunderman@uwaterloo.ca}
\affiliation{%
 The Institute for Quantum Computing, University of Waterloo, Waterloo, Ontario, N2L 3G1, Canada
}%
\affiliation{Department of Physics and Astronomy, University of Waterloo, Waterloo, Ontario, N2L 3G1, Canada}

\date{\today}

\begin{abstract}
Transformations which convert between Fermionic modes and qubit operations have become a ubiquitous tool in quantum algorithms for simulating systems. Similarly, collections of Pauli operators might be obtained from solutions of non-local games and satisfiability problems. Drawing on ideas from entanglement-assisted quantum error-correcting codes and quantum convolutional codes, we prove the obtainable lower-bound for the number of qubits needed to represent such Pauli operations which are equivalent and provide a procedure for determining such a set of minimal register Pauli operations.
\end{abstract}

\maketitle

The initial idea for a quantum computer was argued as a tool for simulating quantum systems \cite{feynman2018simulating}. This problem was advanced by Lloyd proposing a method for actually simulating quantum systems by approximating it with products of operators \cite{lloyd1996universal}. Unfortunately such products have very poor error-scaling unless higher order product series are used such as those of Suzuki \cite{suzuki1990fractal,suzuki1991general}. Additionally, this seminal work by Lloyd did not specify what sort of operations would be used in the product formula. The simplest building block operations for quantum devices are the Pauli operators, so if the system can be simulated using these operations the feasibility of utilizing some sort of near-term advantage would be improved, providing a useful algorithm for NISQ (noisy intermediate-scale quantum) devices \cite{preskill2018quantum}. A variety of transformations have been devised for transforming Hamiltonians for quantum systems, such as molecules or field-theories, into Pauli operators with associated weights \cite{bravyi2002fermionic,tranter2018comparison,mcclean2014exploiting,mcclean2020openfermion}. Well controlled qubits with long coherence time are a precious resource, especially in more near-term quantum devices. These all perform reasonably well in this regard, however, this work solves one open question related to these transformations: for a particular collection of Pauli operators obtained from these transformations, or from some other setting, what is the minimal number of qubits needed in order to perform such operations? Our answer is found from quantum error-correction techniques. Additionally, our solution finds the obtainable lower-bound and also provides an efficient methodology for transforming the original collection of Pauli operators into an equivalent collection using the minimal number of qubits. Note that this work has more general applications--wherever Pauli operators appear (from quantum complexity theory to non-local games to fault-tolerance to communication protocols)--but we focus here on more near-term algorithms related to simulating quantum systems.

\section{Background Definitions}

Formally, quantum computers operate by performing unitary operators. The rudimentary operations for qubits are the Pauli operators, which act as follows on the computational basis states $\{|0\rangle,|1\rangle\}$:
\begin{equation}
    X|j\rangle=|j\oplus 1\rangle,\quad Z|j\rangle =(-1)^j|j\rangle,\quad j\in \{0,1\}.
\end{equation}
There is one final non-trivial operator, $Y$, which we take as the product of $X$ and $Z$ \footnote{While $XZ= -ZX$ and $Y=iXZ$, for the sake of our arguments it suffices to note that $Y$ is just some scalar times the product}. \if{false}As matrices these are explicitly:
\begin{equation}
    X=\begin{bmatrix}
    0 & 1\\
    1 & 0
    \end{bmatrix},\ Y=\begin{bmatrix}
    0 & -i\\
    i & 0
    \end{bmatrix},\ Z=\begin{bmatrix}
    1 & 0\\
    0 & -1
    \end{bmatrix}.
\end{equation}\fi
The $n$-fold tensor product of Pauli operators is indicated by $\mathbb{P}^n$.

Although the Pauli operators form a basis, working with these non-commuting operators is rather cumbersome and so we utilize the symplectic representation for the Pauli operators to transform the tool kit from those of non-commuting operator theory to those of linear algebra over $\mathbb{Z}_2$ \cite{nielsen2002quantum,ketkar2006nonbinary}.

\begin{definition}
The binary symplectic representation of a Pauli operator is the homomorphic map, $\phi$, which maps $\mathbb{P}^n\mapsto \mathbb{Z}_2^{2n}$. It takes the $i$-th register in the Pauli and sets the $i$-th position in the binary vector to the power of the $X$ operator at that position and the $i+n$-th position to the power of the $Z$ operator. In short, this performs:
\begin{equation}
    \phi \left(\bigotimes_{t=1}^n X^{a_t}Z^{b_t}\right)=\left(\bigoplus_{t=1}^n a_t\right)\bigoplus \left(\bigoplus_{t=1}^n b_t\right).
\end{equation}
This mapping observes:
\begin{equation}
    \phi(s_i\circ s_j)=\phi(s_i)\oplus \phi(s_j)
\end{equation}
where $s_i$ and $s_j$ are two Paulis in $\mathbb{P}^n$ and $\oplus$ is componentwise vector addition mod two.
\end{definition}

\begin{definition}
The commutator between two Pauli operators $\phi(p_1)=(\vec{a}_1|\vec{b}_1)$ and $\phi(p_2)=(\vec{a}_2|\vec{b}_2)$ is computed as:
\begin{equation}
    \phi(p_1)\odot \phi(p_2)=\vec{a}_1\cdot \vec{b}_2-\vec{b}_2\cdot \vec{a}_2 \mod 2
\end{equation}
This is $0$ if $p_1$ and $p_2$ commute, and $1$ if not (they anti-commute).
\end{definition}

\begin{definition}
For a collection of Pauli operators $\mathcal{P}$ we denote by $rank_\phi(\mathcal{P})$ the number of compositionally independent generators for $\mathcal{P}$, we call this the $\phi$-rank of $\mathcal{P}$.
\end{definition}

The $\phi$-rank of $\mathcal{P}$ is simply the rank of the symplectic representation of $\mathcal{P}$. Since there will be a second rank of importance in this work, we have made this separate distinction.

\section{Auxiliary Results}

In this section we begin by providing an example illustrating the result of this work, then state the main result, following it by proving smaller results in this section before showing the main result in the next.

\begin{example}
Consider the following pair of operators $\{XX,IZ\}$ as a motivating example. These two operators anti-commute and the rank of their symplectic form is $2$, however, a shorter set of Pauli operators which also obey these relations is $X, Z$. So the number of physical registers needed can be reduced for these operators from $2$ to just $1$. This work considers this problem and shows the optimal reduction in the number of registers for collections of Pauli operators so that the same commutation relations are preserved.
\end{example}

\begin{definition}
Supplementary Pauli operators are a collection of Pauli operators, $\mathcal{P}^\cup$, associated with each Pauli in $\mathcal{P}$ such that $\mathcal{P}^\cup\otimes_r \mathcal{P}$ has all operators commuting, where $\otimes_r$ is the tensor product of element $i$ from $\mathcal{P}^\cup$ with element $i$ from $\mathcal{P}$.
\end{definition}

\begin{example}
Consider again the prior example, then we have $\mathcal{P}=\{XX,IZ\}$ and $\mathcal{P}^\cup=\{X,Z\}$ with $\mathcal{P}^\cup\otimes_r \mathcal{P}=\{XXX,ZIZ\}$.
\end{example}

\begin{remark}\label{commeq}
Note that $\phi(\mathcal{P}_i)\odot \phi(\mathcal{P}_j)=\phi(\mathcal{P}^\cup_i)\odot \phi(\mathcal{P}^\cup_j)$ since we are working in the qubit, binary, case. Generally $\phi(\mathcal{P}_i)\odot \phi(\mathcal{P}_j)+\phi(\mathcal{P}^\cup_i)\odot \phi(\mathcal{P}^\cup_j)=0$, with $+$ being standard addition mod the local-dimension value \cite{ketkar2006nonbinary}.
\end{remark}

Given this, we will focus on finding a set of $\mathcal{P}^\cup$ that has a minimal number of registers. This problem is exactly that which entanglement-assisted quantum error-correcting codes needed to solve--finding a set of supplemental Paulis to act on pre-shared entanglement to resolve non-commuting generators for the codes \cite{brun2006correcting,hsieh2008entanglement,bennett1996mixed}. This same problem also appears in determining the number of registers that must be carried over between frames in quantum convolutional codes \cite{houshmand2012minimal}. Fortunately, this problem has been studied and been solved, although such solutions do not provide a method for determining these supplementary Paulis \cite{wilde2008optimal,houshmand2012minimal}. For these works the primary concern was generating Pauli operators that with their supplementary Pauli operators pairwise commuted, whereas here we focus on  ensuring these collections have the same commutation relationships. The universality of this result from EAQECC and quantum convolutional codes is brought forth in this work, showing that the utility of this method extends beyond error-correction but can also play a central role in improving some simulation based quantum algorithms and any place where stabilizers are a powerful tool. Additionally, as it provides a tight lower-bound for the Paulis obtained from a fermionic to qubit transformation, it shows that further reduction is not possible, so other avenues must be considered, such as clique formation \cite{verteletskyi2020measurement, jena2019pauli, shlosberg2021adaptive, paulson2021simulating, van2020circuit, gokhale2019minimizing} or randomized Trotterization \cite{campbell2019random,childs2019faster} or some other method which maps the Hamiltonian to Pauli operators.

Our primary result is then the following Theorem, which we will prove following auxiliary Lemmas and associated Definitions for these Lemmas.

\begin{theorem}\label{mainresult}
Given a collection of Pauli operators with associated weights on $n$ registers, equivalent Pauli operators with the same $\phi$-rank can be efficiently determined which obtain the minimal number of registers needed.
\end{theorem}

In order to reduce this problem, we may consider just a basis subset from the collection of Pauli operators $\mathcal{P}$:

\begin{lemma}\label{bases}
It suffices to find generators for the smallest subgroup that contains all of the given Pauli operators, $\mathcal{P}$.
\end{lemma}

\begin{proof}
Let $\mathcal{P}\subseteq S$ and $s$ be a set of generators for $S$ such that $s\subseteq \mathcal{P}$. Since the homomorphism $\phi$ carries composition as addition and the symplectic product is linear in the $\phi$ representations for Pauli operators, the elements of $\mathcal{P}$ that are not in $s$ can be generated from composition combinations of elements in $s$. Thus to fully specify $\mathcal{P}$, $s$ suffices--the weights associated with each operator must be tracked separately, however, these determine the full set of commutation relations that the collection of Pauli operators must satisfy.
\end{proof}

\begin{remark}
This result holds for a collection of supplementary Pauli operators, $\mathcal{P}^\cup$.
\end{remark}

Note that since we wish to also satisfy the commutation relations also for compositions of the generators, we require $rank_\phi(\mathcal{P})=rank_\phi(\mathcal{P}^\cup)$.

\begin{definition}
The commutation matrix for a collection of Pauli operators $\mathcal{P}$, $M(\mathcal{P})$, has entries $[M(\mathcal{P})]_{ij}=\phi(\mathcal{P}_i)\odot \phi(\mathcal{P}_j)$ for basis members of $\mathcal{P}$.
\end{definition}

From Remark \ref{commeq}, we know that $M(\mathcal{P})=M(\mathcal{P}^\cup)$.

\begin{lemma}\label{invar}
The commutator matrix's rank for a collection of Paulis $\mathcal{P}$, $rank(M(\mathcal{P}))$, is invariant under the choice of basis elements.
\end{lemma}

\begin{proof}
Let us take some note on how composing generators alters the commutator matrix $M$. The composition of generator $i$ onto generator $j$, written as $i\circ j$, results in the simultaneous addition of row $i$ in $M$ to row $j$, as well as the addition of column $i$ to column $j$ in $M$. This is evident from the homomorphism rule of composition being equivalent to addition of vectors as well as addition in the commutators.

Generally, the simultaneous addition row and column $i$ to $j$ is equal to the row addition, then the column addition, then the addition of $[M]_{ii}$ to $[M]_{jj}$. These commutativity matrices are anti-symmetric (including the qudit case) and have diagonals that are all zero, as each operator trivially commutes with itself, which ensures that the recipient diagonal entry is left as zero. Then for these matrices the simultaneous addition of rows and columns is equal to row addition followed by column addition. These are elementary matrix operations so they will preserve the rank.
\end{proof}

\section{Determining a valid set of minimal register operators}

As the commutation relations between a collection of Pauli operators determines how the members of this weighted collection interact with each other, finding any collection of Pauli operators which satisfy the same commutation relations will perform the same operations as the original collection. This means that by processing the collection of Paulis before running the experiment we may reduce the number of registers needed in the experiment, then upon completing said experiment covert the results back in terms of the original operators and in turn the original Hamiltonian. Alternatively, it states that one can use whichever method one desires to determine a collection of Pauli operators which satisfy the commutativity requirements then apply this method to compress those operators into ones which utilize the fewest qubits.

In a prior work it was found that the size of the minimal-memory encoder for a quantum convolutional code $S$ was given by $dim(M(S))-\frac{1}{2}rank(M(S))$ \footnote{Note that the rank is computed over mod $2$, as all operations are performed over this field.} \cite{houshmand2012minimal}, which is equivalent to the number of registers needed in our supplementary Paulis. This formula can be derived by breaking the commutation matrix into its isotropic portion (similar to a null space) and symplectic portions (a direct sum of rank $2$ matrices). While this formula provided the number of qubits needed, it did not provide any method for obtaining the memory operators themselves. Algorithm 1 from \cite{houshmand2012minimal} provides a procedure for shortening the size of the memory operators, effectively removing any additional memory operators needed that may have been included due to composition of delayed generators, however that is not relevant in this work. While a useful tool, these results do not provide a methodology for generating a set of valid memory operators. The following Lemma provides such a constructive method, which is of particular use here. We will use this Lemma in the proof of Theorem \ref{mainresult}.


\begin{lemma}\label{algo}
There is a deterministic, efficient procedure for determining a valid set of supplementary Pauli operators $\mathcal{P}^\cup$ using the minimal number of registers for a given collection of Pauli operators $\mathcal{P}$ with commutation matrix $M$.
\end{lemma}

\begin{proof}
Let $dim(M)=d$. From Lemma \ref{invar} the composition of generators is equivalent to conjugation by elementary row additions. Upon composing generators we also retain the anti-symmetric nature of $M$. Begin with the first nonzero column in $M$ and add rows such that this column is given by $0100\ldots$. The first row will also be in this form. Repeat this procedure now with the second column, which will have a nonzero first entry. Continue in this way, performing Gaussian elimination, until the matrix is, up to permutations, given by:
\begin{equation}
    \Tilde{D}:=\left(\bigoplus_{i=0}^{dim(M)-rank(M)} [0]\right)\bigoplus \left(\bigoplus_{i=0}^{\frac{1}{2}rank(M)} \begin{bmatrix}
    0 & 1\\
    1 & 0
    \end{bmatrix}\right).
\end{equation}
Let $L^{-1}$ be the set of elementary operations performed to transform $M$ to $\Tilde{D}$--so that $M=L\Tilde{D}L^T$. A set of Pauli operators with commutativity matrix $\Tilde{D}$ is given by $\{Z_i\}_{i=1}^{dim(M)-rank(M)}$ and $\{X_i,Z_i\}_{i=dim(M)-rank(M)+1}^{dim(M)-\frac{1}{2}rank(M)}$. To recover a set of Paulis satisfying $M$ with the minimal number of registers, $\mathcal{P}^\cup$, we may simply compose these aforementioned generators' $\phi$ representations with $L$.
\end{proof}

Putting together the auxiliary results from the prior section and Lemma \ref{algo}, we have proven Theorem \ref{mainresult}.

\begin{corollary}
This result holds for qudits as well.
\end{corollary}

Note that this Lemma in essence re-proves the result from \cite{wilde2008optimal,brun2006correcting} in a modified language, circumventing the need to appeal to symplectic vector spaces and the Symplectic Graham-Schmidt Orthogonalization Procedure (SGSOP), instead working with \textit{nearly} standard Gaussian elimination. This also means that this is likely the computationally fastest routine that can exist for determining the $\mathcal{P}^\cup$ as well as the equivalent problems of finding the priorly shared entanglement in EAQECC and the memory operators in quantum convolutional codes. Additionally this proof provides a methodology for determining a satisfying set of Pauli operators, a piece that was lacking in prior works, as the reduction occurs through repeated row additions (compositions of the generators to be found) but does not require any Clifford operations.

\if{false}
\begin{proof}
The case of dimension $d=2$ is handled trivially: if rank $2$, then $\{X,Z\}$, else if rank zero then $\{Z_1, Z_2\}$ suffice.

Let $\text{dim}(M)=d\geq 3$. This proof is built by first considering the two extremes for the rank of the memory matrix. First, let the rank be zero, then the number of qubits in the memory operators is $d$. In this case, a sufficient set of memory operators is formed by $Z_i$, with $i\in [d]$, which is a single $Z$ Pauli acting on qubit $i$.

Next, consider the case of a nearly full rank memory matrix such that each column has at least a single $1$. Further, let us suppose that the matrix is hollow and tridiagonal with $1$ along the secondary diagonals. In this case the number of qubits is given by $\lceil \frac{d}{2}\rceil$. For this matrix it must be the case that $m_1$ and $m_2$ (the first two Pauli operators in $\mathcal{P}^\cup$) don't commute, so let's choose: $m_1=X_1$ and $m_2=Z_1Z_2$. It must also be the case that $m_2$ and $m_3$ don't commute. Picking $m_3=X_2$ suffices. We may continue in this way, alternating between a single $X$ operator and a pair of adjacent $Z$ operators to form a complete set of operators with this memory matrix. The only exception that must be made is that the last $Z$ operator must consist of a single $Z$ operator on the last qubit.

Now bring in Lemma \ref{invar}. Consider the action of composing memory operators. This is equivalent to adding that row \textit{and} column within the symplectic matrix to the other row \textit{and} column. And so with this hollow, tridiagonal matrix (with whatever binary rank) we may generate any memory matrix with the same rank. In effect any symmetric matrix with the same rank may be written as $L\Tilde{D} L^T$, for some matrix $L$ formed from the product of elementary matrix operations--only row swaps and row additions are needed here.

Lastly, for the general case, any memory matrix may be written in the form:
\begin{equation}
    M=\begin{bmatrix}
    0 & 0\\
    0^T & L\Tilde{D}L^T
    \end{bmatrix},
\end{equation}
where we have just reordered the rows (with their columns) into this form (which does not alter the operators, just their labels) so that all zero columns and rows are as the beginning of the matrix. Let's say the upperleft $0$ matrix has dimensions $b\times b$, the other $0$ has dimensions $b\times a$, and $L \Tilde{D}L^T$ has dimensions $a\times a$--with $a+b=d$.

With this representation, we select as the first $b$ memory operators $Z_i$. For the remaining $a$ operators we now select operators from the alternating set, composed over the qubits following the first $b$ such that we obtain the $L\Tilde{D}L^T$ block. This is achieved by taking these alternating operators and composing them by multiplication by $L$. From here, so long as the labelling is kept track of, the memory operators are obtained. 
\end{proof}

This result, on its own, is not surprising as this is equivalent to the fundamental theorem of symplectic algebra, which formed the basis for the proof of the number of memory qubits needed and entanglement qubits needed for entanglement-assisted quantum error-correcting codes (EAQECC). While the fundamental theorem of symplectic algebra was directly used in the case of showing EAQECC, the analogy between convolutional code memories and the priorly shared entanglement in EAQECC puts this result into clearer context. In the language of EAQECC this states that the generators, upon using "canonical codes", can be broken into an isotropic portion and a symplectic portion \cite{hsieh2008entanglement}. These portions, in quantum convolutional codes, are equivalent to the memory operators which are a single $Z$ operator and those which are linear combinations of the paired set. However, while the ability to generate the memory operators is not surprising, this provides a simple constructive method for generating the memory operators. When this result is recast within the problem we are solving at hand, it provides a perfectly obtainable lower-bound on the number of registers needed in the experimental device, although loses the clear interpretation of an isotropic subspace and a symplectic subspace. It is worth noting that these minimal register Pauli operators are also obtained from simple compositions of "elementary" Pauli operators. 

\begin{proof}[Proof of Theorem \ref{mainresult}]
Begin by reducing $\mathcal{P}$ to a generating set of operators $\mathcal{P}'$, which suffices as per Lemma \ref{bases} to find the complete set collection of operators. We will find a collection of supplementary Paulis to $\mathcal{P}'$, $\mathcal{P}'^\cup$. Since the commutator matrices for $\mathcal{P}'$ and $\mathcal{P}'^\cup$ are the same--$M(\mathcal{P}')=M(\mathcal{P}'^\cup)$--solving either will return equivalent Pauli operators. The number of registers needed to satisfy the commutation requirements of $M(\mathcal{P}'^\cup)$ is given by $dim(M)-\frac{1}{2} rank(M)$ (which is the same as $rank_\phi(\mathcal{P}^\cup)-\frac{1}{2}rank(M(\mathcal{P}^\cup))$), and the constructive proof from Lemma \ref{algo} will generate such a minimal collection of Pauli operators. Assuming all labels of Paulis are tracked, the weights of the supplementary Paulis will be those of the original Paulis, completing the proof and construction.
\end{proof}
\fi

To close this section we provide an example illustrating this result:
\begin{example}
We will take:
\begin{multline}
    \mathcal{P}=\{ZYZZXZXYXI, IIXYYZIYYY,\\ IYXIXIXYZY, YZZIXZZXYI, ZZZYIXYXXZ,\\ XZZIXIIXIZ, XXIYYIIYIX, IXIXIYIIYI\}.
\end{multline}
There are $8$ operators on $10$ qubits. From the $\phi$ representation, we know that we will need between $4$ and $8$ registers (inclusive) since those are the minimal rank while using both the $X$ and $Z$ components of the representation and the maximal rank when only needing one component (all operators commute). For these Paulis the commutator matrix is given by:
\begin{equation}
    M(\mathcal{P})=\begin{bmatrix}
    0 & 0 & 0 & 1 & 1 & 1 & 0 & 0\\
    0 & 0 & 0 & 1 & 1 & 0 & 1 & 0\\
    0 & 0 & 0 & 1 & 0 & 0 & 1 & 0\\
    1 & 1 & 1 & 0 & 0 & 1 & 0 & 0\\
    1 & 1 & 0 & 0 & 0 & 1 & 0 & 0\\
    1 & 0 & 0 & 1 & 1 & 0 & 0 & 1\\
    0 & 1 & 1 & 0 & 0 & 0 & 0 & 1\\
    0 & 0 & 0 & 0 & 0 & 1 & 1 & 0
    \end{bmatrix}
\end{equation}
This matrix has (binary) rank $6$, which from our formula says that we should need exactly $5$ registers to represent this set of commutation relations.

We perform the following operations to put the matrix into our pseudodiagonal form: add $5$ to $4$, then add $4$ to $7$, then add $7$ to $6$, then add $5$ to $6$, then add $8$ to $2$, then add $6$ to $1$, then add $6$ to $2$:
\begin{equation}
    \begin{bmatrix}
    0 & 0 & 0 & 0 & 0 & 0 & 0 & 0\\
    0 & 0 & 0 & 0 & 0 & 0 & 0 & 0\\
    0 & 0 & 0 & 1 & 0 & 0 & 0 & 0\\
    0 & 0 & 1 & 0 & 0 & 0 & 0 & 0\\
    0 & 0 & 0 & 0 & 0 & 1 & 0 & 0\\
    0 & 0 & 0 & 0 & 1 & 0 & 0 & 0\\
    0 & 0 & 0 & 0 & 0 & 0 & 0 & 1\\
    0 & 0 & 0 & 0 & 0 & 0 & 1 & 0
    \end{bmatrix}.
\end{equation}
\if{false} For ease we have reduced it to the case of a single nonzero entry per column, aside from the $0$ columns. We next add $3$ to $5$, then add $8$ to $6$, to obtain:
\begin{equation}
    \begin{bmatrix}
    0 & 0 & 0 & 0 & 0 & 0 & 0 & 0\\
    0 & 0 & 0 & 0 & 0 & 0 & 0 & 0\\
    0 & 0 & 0 & 1 & 0 & 0 & 0 & 0\\
    0 & 0 & 1 & 0 & 1 & 0 & 0 & 0\\
    0 & 0 & 0 & 1 & 0 & 1 & 0 & 0\\
    0 & 0 & 0 & 0 & 1 & 0 & 1 & 0\\
    0 & 0 & 0 & 0 & 0 & 1 & 0 & 1\\
    0 & 0 & 0 & 0 & 0 & 0 & 1 & 0
    \end{bmatrix}.
\end{equation}
\fi
This matrix is now in the reduced form and can be satisfied by $\{Z_1,Z_2,X_3,Z_3,X_4,Z_4,X_5,Z_5\}$. Then matrix $L$ which we conjugated by to obtain this form is given by the inverse of the product of these elementary row additions, providing:
\begin{equation}
    L=\begin{bmatrix}
    1 & 0 & 0 & 0 & 0 & 1 & 0 & 0\\
    0 & 1 & 0 & 0 & 0 & 1 & 0 & 1\\
    0 & 0 & 1 & 0 & 0 & 0 & 0 & 0\\
    0 & 0 & 0 & 1 & 1 & 0 & 0 & 0\\
    0 & 0 & 0 & 0 & 1 & 0 & 0 & 0\\
    0 & 0 & 0 & 0 & 1 & 1 & 1 & 1\\
    0 & 0 & 0 & 1 & 0 & 0 & 1 & 0\\
    0 & 0 & 0 & 0 & 0 & 0 & 0 & 1
    \end{bmatrix}.
\end{equation}

This satisfies $M(\mathcal{P})=L\Tilde{D}L^T$, so then a minimal register collection of Pauli operators can be given by:
\begin{equation}
    \mathcal{P}^\cup=\{Z_1Z_4,Z_2Z_4Z_5,X_3,Z_3X_4,X_4,Y_4Y_5,Z_3X_5,Z_5\}
\end{equation}
\if{false}
\begin{multline}
    \{Z_1Z_4Z_5Z_5,Z_2Z_4Z_5,X_3,X_3Z_3Z_4X_4,\\
    X_3X_4,X_3X_4Z_4Z_5X_5Z_5,Z_3Z_4X_5,Z_5\},
\end{multline}
which when simplified and expressed in terms of the full qubit Paulis gives:
\begin{multline}
    \mathcal{P}^\cup=\{Z_1Z_4,Z_2Z_4Z_5,X_3,Y_3Y_4,\\
    X_3X_4,X_3Y_4X_5,Z_3Z_4X_5,Z_5\},
\end{multline}
\fi
These Pauli operators on $5$ registers, with weight at most $3$, satisfy the same commutation matrix $M$ for $\mathcal{P}$, so these form a minimal register set of equivalent Pauli operators.
\end{example}

\section{Conclusion}

Pauli operators form the most fundamental building blocks in quantum computations, error-correction, and algorithms. In this work we have couched the discussion of the possible uses of this result in terms of near-term algorithms, however, it can be of rather general use, since knowing the minimal number of registers to represent a collection of commutation relations could be useful for a variety of settings. This work poses the question of what the minimal number of registers needed to represent a particular commutation relation is, then proceeds to solve this problem. This work proves an optimal register compression algorithm for a given set of commutation relations. In particular, this work imports a central idea from entanglement-assisted quantum error-correction and quantum convolutional codes into the broader framework of Pauli operators for any purpose. From the subset of examples considered thus far on small molecules, it seems that the Bravyi-Kitaev transformation achieves this optimal register usage \footnote{Andrew Jena, private communications}\cite{bravyi2002fermionic}. Is this always true and if so, can this be proven? Another future direction to carry this work would be to see whether this method, with appropriate extra considerations, can be used to prove lower-bounds on the weights of Pauli operators needed in order to obtain a collection of commutation relations, and how close a deterministic procedure can get to this minimal weight. Another important avenue to consider would be how the results are impacted by connectivity constraints of a physical device.

\section*{Acknowledgments}

We thank Andrew Jena for helpful discussions, David Cory for feedback, and Mark Wilde for convincing the author to look more closely at EAQECC and quantum convolutional codes.

\bibliographystyle{unsrt}
\phantomsection  
\renewcommand*{\bibname}{References}

\bibliography{main}

\end{document}